\documentclass[UKenglish, cleveref]{lipics-v2019}

\newcommand{\lipics}{}
\usepackage{ifthen}
\usepackage{comment}
\usepackage{complexity}
\makeatletter
\ifthenelse{\not{\isundefined{\lipics}}}
{
\makeatletter
\let\c@author\relax
\makeatother
}{}
\usepackage[
style=numeric,
giveninits=true, 
maxbibnames=99, 
maxcitenames=2, 
natbib=true, 
labelnumber, 
defernumbers=true,  
backend=biber, 
sortcites=true,
doi=false,isbn=false,url=false
]{biblatex}
\renewbibmacro{in:}{%
  \ifentrytype{article}{}{\printtext{\bibstring{in}\intitlepunct}}
}
\let\oldep\everypar  \newtoks\everypar  \oldep{\the\everypar\looseness=-1}

\ifthenelse{\isundefined{\lipics}}
{
\usepackage[protrusion=true,expansion=true]{microtype}

\usepackage[colorlinks=true,bookmarks=false]{hyperref}
\usepackage[nameinlink,capitalise]{cleveref}
\usepackage[small,bf]{caption}
\usepackage{graphicx}
\usepackage{subcaption}
\usepackage[british]{babel}
\usepackage{paralist}
}
{}
\usepackage{xcolor}
\definecolor{darkblue}{rgb}{0,0,0.45}
\definecolor{darkred}{rgb}{0.6,0,0}
\definecolor{darkgreen}{rgb}{0.13,0.5,0}
\hypersetup{colorlinks, linkcolor=darkblue, citecolor=darkgreen,
urlcolor=darkblue}

\ifthenelse{\isundefined{\llncs}}{
  \usepackage{amsthm}
}

\usepackage{amsmath,amsfonts,amssymb}
\usepackage{mathtools}
\usepackage{mathrsfs}
\usepackage{wrapfig}
\usepackage{xspace}
\usepackage{csquotes}
\usepackage{multirow}

\usepackage[ruled, linesnumbered, vlined]{algorithm2e}
\crefname{algocf}{Algorithm}{Algorithms}

\usepackage{footmisc}

\usepackage{aliascnt}
\ifthenelse{\isundefined{\llncs}}{
  \theoremstyle{plain}
}{}
\newaliascnt{lem}{thm}
\newaliascnt{crl}{thm}
\newaliascnt{dfn}{thm}
\newaliascnt{clm}{thm}
\newaliascnt{prop}{thm}
\newaliascnt{hyp}{thm}
\newaliascnt{obs}{thm}
\newaliascnt{fact}{thm}
\newtheorem{fact}[theorem]{Fact}
\newtheorem{reduction}{Reduction rule}
\ifthenelse{\isundefined{\llncs}}{
  \theoremstyle{definition}
}{}

\aliascntresetthe{lem}
\aliascntresetthe{crl}
\aliascntresetthe{prop}
\aliascntresetthe{dfn}
\aliascntresetthe{clm}
\aliascntresetthe{hyp}
\aliascntresetthe{obs}
\aliascntresetthe{fact}
\crefname{thm}{Theorem}{Theorems}
\crefname{lem}{Lemma}{Lemmas}
\crefname{obs}{Observation}{Observations}
\crefname{fact}{Fact}{Facts}

\newcommand{\eps}{{\varepsilon}}

\usepackage{environ,xstring}
\newif\iflabel
\newif\ifdbs
\newif\ifamp
\NewEnviron{doitall}{%
  \noexpandarg
  \expandafter\IfSubStr\expandafter{\BODY}{\label}{\labeltrue}{\labelfalse}%
  \expandafter\IfSubStr\expandafter{\BODY}{\\}{\dbstrue}{\dbsfalse}%
  \expandafter\IfSubStr\expandafter{\BODY}{&}{\amptrue}{\ampfalse}%
  \iflabel\def\doitallstar{}\else\def\doitallstar{*}\fi
  \ifdbs
    \ifamp
      \def\doitallname{align}%
    \else
      \def\doitallname{multline}%
    \fi
  \else
    \def\doitallname{equation}
  \fi
  \begingroup\edef\x{\endgroup
    \noexpand\begin{\doitallname\doitallstar}%
    \noexpand\BODY
    \noexpand\end{\doitallname\doitallstar}%
  }\x
}
\def\[#1\]{\begin{doitall}#1\end{doitall}}

\usepackage{framed}
\newcommand{\executeiffilenewer}[3]{%
\ifnum\pdfstrcmp{\pdffilemoddate{#1}}%
{\pdffilemoddate{#2}}>0%
{\immediate\write18{#3}}\fi%
}

\newcommand{\defproblem}[3]{
  \vspace{1mm}
\noindent\fbox{
  \begin{minipage}{0.96\textwidth}
  \begin{tabular*}{\textwidth}{@{\extracolsep{\fill}}lr} #1 \\ \end{tabular*}
  {\bf{Input:}} #2  \\
  {\bf{Output:}} #3
  \end{minipage}
  }
  \vspace{1mm}
}

\nolinenumbers
\bibliography{edge-clique-partition}

\newcommand{\calS}{\mathcal{S}}

\newcommand{\calC}{\mathcal{C}}

\newcommand{\wild}{\star}

\newcommand{\zs}{(\mathbb{Z}_{\ge 0}\cup\{\wild\})}
\newcommand{\zplus}{\mathbb{Z}^{+}}
\newcommand{\bin}{\{0,1\}}
\newcommand{\binkk}{\{0,1\}^{k\times k}}
\newcommand{\bigO}{\mathcal{O}}
\newcommand{\es}{\triangleq}
\newcommand{\basis}{\tilde{B}}

\newcommand{\RR}{\mathbb{R}}
\newcommand{\N}{\mathbb{N}}
\newcommand{\B}{\mathbb{B}}
\newcommand{\Bnk}{\mathbb{B}^{n\times k}}

\newcommand{\nnr}{R}
\newcommand{\zerov}{\mathbf{0}}

\newcommand{\ebas}{\mathtt{ExtendBasis}}
\newcommand{\ecp}{\textsc{ECP}}

\title{Fixed-Parameter Tractability of the \texorpdfstring{\newline}{} Weighted Edge Clique 
Partition Problem}

\titlerunning{Fixed-Parameter Tractability of the Weighted Edge Clique Partition 
Problem}

\author{Andreas~Emil~Feldmann}{Department of Applied Mathematics, Charles 
University in Prague, Czechia}
{feldmann.a.e@gmail.com}
{https://orcid.org/0000-0001-6229-5332}
{Czech 
Science Foundation GA{\v C}R (grant \#19-27871X), and Center for Foundations of 
Modern Computer Science (Charles Univ.\ project UNCE/SCI/004).}

\author{Davis~Issac}{Hasso Plattner Institute, Potsdam, 
Germany.}{davis.issac@hpi.de}{https://orcid.org/0000-0001-5559-7471}
{the major part of the work 
was done when this author was a postdoctoral researcher at Charles University, 
Prague, Czechia. He was funded by Charles Univ.\ project UNCE/SCI/004.}

\author{Ashutosh~Rai}{Department of Applied Mathematics, Charles University in 
Prague, Czechia}{ashuthosh@kam.mff.cuni.cz}{https://orcid.org/0000-0003-2429-750X}
{supported by Center for Foundations of Modern Computer Science (Charles Univ. project UNCE/SCI/004).}
\authorrunning{A.~E.~Feldmann, D.~Issac, A.~Rai}

\Copyright{Andreas~Emil~Feldmann, Davis~Issac, Ashutosh~Rai}

\ccsdesc[500]{Theory of computation~Fixed parameter tractability}

\keywords{%
Edge Clique Partition,
fixed-parameter tractability,
kernelization
}%

\begin{document}

\maketitle

\begin{abstract}
We develop an FPT algorithm and a bi-kernel for the Weighted Edge Clique 
Partition (WECP) problem, where a graph with $n$ vertices and integer edge 
weights is given together with an integer $k$, and the aim is to find $k$ 
cliques, such that every edge appears in exactly as many cliques as its weight. 
The problem has been previously only studied in the unweighted version called 
Edge Clique Partition (ECP), where the edges need to be partitioned into $k$ 
cliques. It was shown that ECP admits a kernel with~$k^2$ vertices [Mujuni and 
Rosamond, 2008], but this kernel does not extend to WECP. The previously fastest 
algorithm known for ECP has a runtime of $2^{\bigO(k^2)}n^{O(1)}$ [Issac, 2019]. 
For WECP we develop a bi-kernel with $4^k$ vertices, and an algorithm with 
runtime $2^{\bigO(k^{3/2}w^{1/2}\log(k/w))}n^{O(1)}$, where $w$ is the maximum 
edge weight. The latter in particular improves the runtime for ECP 
to~$2^{\bigO(k^{3/2}\log k)}n^{O(1)}$. 
\end{abstract}

\section{Introduction}

Problems that aim to cover a graph by a small number of cliques have a long 
history and have been studied extensively in the 
past (see e.g.~\cite{chalermsook2014nearly, chandran2017parameterized, 
cygan2016known, gramm2006data, ma1988complexity, mujuni2008parameterized, 
fleischer2010edge, fleischner2009covering}). For these types 
of problems we are given a graph~$G$ and an integer $k$, and the tasks include 
to either cover or partition the edges or the vertices of $G$ using at most $k$ 
cliques or bicliques (i.e., complete bipartite graphs). 
Plenty of applications 
exist in both theory~\cite{roberts1985applications} and 
practice, e.g., in computational 
biology~\cite{blanchette2012clique,figueroa2004clustering}, compiler 
optimization~\cite{rajagopalan2000handling}, language 
theory~\cite{gruber2007inapproximability}, and database 
tiling~\cite{geerts2004tiling}. In this paper, we study the variant called the 
Edge Clique Partition (\ecp) problem,
defined as follows.

\defproblem{\large{ECP} (Edge Clique Partition)}{a graph $G$, a positive integer $k$}{a partition of the edges of $G$ into $k$ cliques (if it exists, otherwise output NO)}

{\ecp} has applications in 
computational biology~\cite{blanchette2012clique,figueroa2004clustering}.
{\ecp} is known to be NP-hard even in $K_4$-free graphs and chordal 
graphs~\cite{ma1988complexity}, and together with~\cite{kann1991maximum}, the 
reductions of~\cite{ma1988complexity} imply APX-hardness. To circumvent 
these hardness results, we focus on {\bf parameterized 
algorithms} (see \cite{cygan2015parameterized} for the basics). 
More specifically, we focus on FPT algorithms for the natural parameter~$k$, i.e., 
the number of cliques. \citet{fleischer2010edge} show that on planar graphs, ECP 
can be solved in~$\bigO^*(2^{96\sqrt{k}})$ time\footnote{the $\bigO^*$-notation 
hides polynomial factors in input size}. They also generalized the result to $d$-degenerate graphs, giving an algorithm with $\bigO^*(2^{dk})$ runtime, which is linear for bounded-degeneracy graphs. For $K_4$-free graphs, 
\citet{mujuni2008parameterized} gave an algorithm with a runtime\footnote{in 
\cite{mujuni2008parameterized} the runtime was mistakenly reported as 
$\bigO^*(k^{(k+3)/2)})$, cf.~\cite{fleischer2010edge}.} of 
$\bigO^*((\frac{k+3}{2})^k)=\bigO^*(2^{\bigO(k\log k)})$, which was improved by 
\citet{fleischer2010edge} to $\bigO^*((\sqrt{k}/3)^k)$ and even 
$\bigO^*((64c)^k)$ for some large (unspecified) constant $c$. Hence, also for 
these graphs an exponent linear in $k$ is possible, albeit with a very large 
base. On the other hand, the algorithm of \citet{mujuni2008parameterized} 
for $K_4$-free graphs has been empirically shown~\cite{wu2011research} to be 
rather efficient, even though it ``only'' comes with a near-linear exponent of 
$\bigO(k\log k)$. 

\citet{mujuni2008parameterized} showed  ECP 
is FPT in $k$ for general graphs, 
by giving a {\bf kernel} (see \cite{cygan2015parameterized} for definition) of size $k^2$.
However, no algorithms with (near-)linear dependence 
on $k$ in the exponent are known for ECP. The fastest algorithm so far is given by 
\citet[Theorem~3.10]{Issac_2019} and runs in~$\bigO^*(2^{2k^2+k\log_2 k+k})$ 
time, i.e., the exponent is quadratic in $k$. This algorithm is an adaptation of 
an algorithm by \citet{chandran2017parameterized} for the Biclique 
Partition problem (where we want to partition the edges 
into $k$ bicliques) in bipartite graphs. In contrast, the best runtime lower bound known for ECP 
only excludes a sub-linear dependence on $k$ in the exponent: there is no 
$2^{o(k)}n^{O(1)}$ time algorithm for ECP assuming the Exponential Time Hypothesis (ETH). 
This follows due to a $2^{o(n)}$ lower bound for 3-Dimensional Matching 
\cite{jansen2016bounding} under ETH, and a reduction from Exact 3-Cover (which 
is a generalization of 3-Dimensional Matching) to ECP by 
\citet{ma1988complexity}. An obvious open problem arising here is to close the 
gap between the upper and lower bounds on the runtime for ECP.
Our main contribution is to show that for general graphs the exponent of the 
runtime for ECP can be significantly lowered from $\bigO(k^2)$ 
to $(k^{3/2}+\bigO(k))\log k$. 
\begin{theorem}
\label{thm:ecp-algo}
ECP has an algorithm running in $(2e)^{(k^{3/2}+\bigO(k))\log_{2e}(k)}+\bigO(n^2\log n)$ time. 
\end{theorem}

In fact, our algorithm solves a more general 
problem that we call the {\bf Weighted Edge Clique Partition} (WECP) problem defined as follows:

\defproblem{\large{WECP} (Weighted Edge Clique Partition)}{a graph $G$, a weight function $w_e:E(G)\rightarrow \zplus$, and a positive integer $k$}{a set of at most $k$ cliques such that each edge appears in exactly as many 
cliques as its weight (if it exists, otherwise output NO)}

Note that WECP is equivalent to ECP on a multigraph, by 
taking the weights as the edge multiplicities. The WECP problem also has 
applications in computational biology, specifically in the inference of gene pathways  
from gene co-expression data~\cite{Blair}. Thus 
developing efficient algorithms for WECP is of practical relevance. It was not 
known till now whether WECP is even FPT; in particular, the known FPT algorithms for 
ECP do not extend to WECP. 
The reason is that the first step of these algorithms is to run the kernelization algorithm but for WECP, no $f(k)$-kernel for any computable function $f$ is known. 
This is in contrast with the $k^2$-kernel of ECP and also a $3^k$-kernel of the very similar Biclique Partition problem by \citet{fleischner2009covering}.
We first show a so-called \emph{bi-kernel} with $4^{k}$ vertices for WECP that can be computed in polynomial 
time. That is, the kernel is for an even more general problem that we call the {\bf 
Annotated Weighted Edge Clique Partition} (AWECP) problem, defined as follows.

\defproblem{\large{AWECP} (Annotated Weighted Edge Clique Partition)}{a graph $G$, edge-weights $w_e:E(G)\rightarrow \zplus$, a special set of vertices $W\subseteq V(G)$, vertex weights $w_v:W\rightarrow \zplus$, and a positive integer $k$}{a set of at most $k$ cliques such that each edge $e$ appears in exactly as many 
cliques as its edge-weight, and each vertex in $W$ appears in exactly as many 
cliques as its vertex-weight (if such $k$ cliques exist, otherwise output NO)}

Note that WECP is exactly the special case of AWECP when 
$W$ is empty. We give a kernel for AWECP as follows.

\begin{theorem}\label{thm:kernel}
	AWECP has a kernelizaiton algorithm that runs in $\bigO(n^2\log n)$ time and 
outputs a kernel having at most $4^{k}$ vertices and encoding length~$\bigO(16^k\log 
k)$ bits.
\end{theorem}

Then we proceed to give the first FPT algorithm for WECP, which also implies the improved algorithm for ECP.

\begin{theorem}
\label{thm:algo}
WECP with the edge weights upper bounded by some value $w$ has an algorithm running in
$(2e)^{(k^{3/2}w^{1/2}+\bigO(k))\log_{2e}(k/w)}+\bigO(n^2\log n)$ time. \footnote{$n$ is always the number of vertices of the input graph unless otherwise mentioned.}
\end{theorem}

Note that \cref{thm:algo} implies an FPT algorithm for WECP when parameterized 
by $k$ as $w\le k$ for any YES-instance.
Also, \cref{thm:ecp-algo} follows from \cref{thm:algo} by setting $w=1$.

\subsection{Our techniques}\label{sec:techniques}

Our approach is based on the work of \citet{chandran2017parameterized}, who 
solve the Bipartite Biclique Partition problem using linear algebraic techniques: we 
express AWECP as a low-rank matrix decomposition problem. For this we allow 
matrices to have wildcard entries in the diagonal that will be denoted 
by~$\wild$. For values $a$ and $b$, we write $a\es b$ if and only if either $a=b$, or 
$a=\wild$, or~$b=\wild$. For two matrices $A$ and $B$, we write $A\es B$ if and 
only if $A_{i,j}\es B_{i,j}$ for all~$i,j$. We say that a matrix $B$ is a 
\emph{Binary Symmetric Decomposition (BSD)} of matrix~$A$ if $BB^T\es A$ and $B$ 
is binary. Further, the matrix $B$ is called a rank-$k$ BSD of $A$ if it is a BSD of $A$ and has at most $k$ columns.
We define the {\bf Binary Symmetric Decomposition with Diagonal Wildcards}
(BSD-DW) problem as follows

\defproblem{\large{BSD-DW} (Binary Symmetric Decomposition with Diagonal Wildcards)}
{an integer non-negative symmetric matrix $A\in 
\zs^{n\times n}$ such that the wildcards $\wild$ appear only in the diagonal, 
and an integer~$k$} 
{a \emph{rank-$k$ BSD} of $A$ (if it exists, otherwise output NO)} 

We prove (in \cref{lem:equiv}) that AWECP and BSD-DW are equivalent. 
Moreover, each column of $B$ (solution to BSD-DW) corresponds to a clique (in the solution to AWECP), i.e. the rows that have a $1$ in the $j$-th column correspond to the vertices that are in the $j$-th clique.
Due to this, we will index the rows and columns of $A$ with vertices, the rows of $B$ 
with vertices and the columns of $B$ with integers from~$[k]$, that correspond 
to the $k$ cliques. Moreover, we will be fluently switching between the contexts 
of edge partitionings of graphs (AWECP), and matrix decomposition (BSD-DW).

In \cref{sec:kernel} we prove that there is a kernel for AWECP with 
$4^k$ vertices. For this we group the vertices into equivalence classes (that we 
call \emph{blocks}) of twin vertices \footnote{our notion of twins is slightly different than 
the usual one in literature}. If a block has size more than~$2^k$, we show that they 
can be reduced and represented by one vertex. For this reduction rule, we need 
to specify how often the representative vertex needs to be covered by cliques. 
Thus, even if the input is an instance of WECP, the kernel we compute will be 
annotated, i.e., it will be an instance of AWECP. The $4^k$ bound on the kernel 
size follows then by giving a $2^k$ upper bound on the number of blocks for a 
YES instance. Since the edge weights and vertex weights for vertices in $W$ 
cannot exceed $k$ if there is a solution with at most $k$ cliques, a kernel with 
at most $4^k$ vertices can be encoded using $\bigO(\binom{4^k}{2}\log k)$ bits, 
and so \cref{thm:kernel} follows.

To obtain \cref{thm:algo}, we first compute a kernel using \cref{thm:kernel} as the
first step of the algorithm. Our algorithm will 
solve the more general AWECP problem. As in the algorithm of \citet{chandran2017parameterized} 
(where a different low-rank matrix decomposition problem is solved), the main 
idea of our algorithm is to guess a row basis for a rank-$k$ BSD~$B$, and then 
fill the remaining rows of $B$ one by one independent of each other. However we need to refine the techniques of 
\mbox{\citet{chandran2017parameterized}} in order to obtain our runtime 
improvement. In particular, there are two reasons why the algorithm 
in~\cite{chandran2017parameterized} has a quadratic dependence on $k$ in the 
exponent: first, to guess a basis of rank~$k$, they need to guess $k$ binary 
vectors of length $k$ each, which takes $\bigO(2^{k^2})$ time. But also, they 
need to guess the $k$ row basis indices of $B$,
for which there are $\binom{m}{k}$ possibilities if the matrix has $m$ 
rows. Since for Biparitie Biclique Partition there is a kernel where $m\leq 2^k$ 
\cite{fleischner2009covering}, this adds another factor of $\bigO(2^{k^2})$ to 
the runtime. 

To circumvent these two runtime bottlenecks, in \cref{sec:algorithm} we devise 
an algorithm that gets around guessing the row indices of the basis of the solution 
matrix $B$. Instead of guessing the whole basis, we add a row to the basis only when the current basis cannot \emph{take care} of that row. While this makes our algorithm more involved than the one by 
\citet{chandran2017parameterized}, it means that the only bottleneck left is 
guessing the basis entries. For BSD-DW we can show that a basis with only 
$k^{3/2}w^{1/2}+k$ ones exists, which follows from the well-studied 
Zarankiewicz problem~\cite{nikiforov2010contribution}. This bound on the 
structure of the basis then implies \cref{thm:algo}.

Since the only bottleneck, which prevents our algorithm from having near-linear 
dependence on $k$ in the exponent of the runtime, is the step that guesses the 
entries of the basis for the solution matrix~$B$, a natural question is 
whether our upper bound of $k^{3/2}w^{1/2}+k$  of the number of ones is (asymptotically) tight. In 
\cref{sec:fpp} we show that this is indeed tight (at least for the unweighted case) by proving the following theorem:
\begin{theorem}
	\label{thm:ones}
For every prime power $N$ and $k=N^2+N$ , there is a matrix $A\in \bin^{(k+1) \times (k+1)}$ such that there is a rank-$k$ BSD for $A$ and every row basis of every rank-$k$ BSD of $A$ has $\Theta(k^{3/2})$ ones.
\end{theorem}
While this does not give 
a runtime lower bound in general, it implies that in order to speed up our algorithm for ECP using a better enumeration of the potential basis matrices, one needs to use some property other than a bound on the number of ones. 
The tight instances are obtained via the well-known Finite Projective Planes.

\subsection{Related results}\label{sec:related}

We now survey some results for ECP and related problems, apart from those 
mentioned above. For ECP, it is also known that the problem is solvable in 
polynomial time on cubic graphs~\cite{fleischer2010edge}. The problem of 
partitioning the vertices instead of the edges into $k$ cliques is equivalent 
to $k$-colouring on the complement graph, which is well-known to be NP-hard even 
for $k=3$. Similarly, when the vertices need to be partitioned into bicliques or 
covered by bicliques, \citet{fleischner2009covering} proved NP-hardness for any 
constant $k\geq 3$.

\emph{Covering} the edges of a graph by cliques or bicliques turns out to be 
generally harder than \emph{partitioning} the edges. For the Edge Clique Cover problem, 
a kernel with $2^k$ vertices was shown by \citet{gramm2006data}, which 
results in a double-exponential time FPT algorithm when solving the kernel by 
brute-force. \citet{cygan2016known} showed that this is essentially best 
possible, as under ETH no $2^{2^{o(k)}}n^{O(1)}$ time algorithm exists for Edge 
Clique Cover and no kernel of size $2^{o(k)}$ exists unless $\P =\NP$. 
Similarly, for the Biclique Cover problem, where edges of a general graph need 
to be covered by bicliques, \mbox{\citet{fleischner2009covering}} gave a 
kernel with $3^k$ vertices, and for the Bipartite Biclique Cover problem they 
gave a kernel with $2^k$ vertices in each bipartition. These kernels naturally 
imply double-exponential time algorithms. 
\mbox{\citet{chandran2017parameterized}} proved that for Bipartite Biclique 
Cover, under ETH no $2^{2^{o(k)}}n^{O(1)}$ time algorithm exists, and unless $\P 
=\NP$ no kernel of size $2^{o(k)}$ exists.

\citet{chalermsook2014nearly} showed that for the Biclique 
Cover problem, it is NP-hard to compute an $n^{1-\eps}$-approximation  
for any~$\eps>0$
\footnote{The paper wrongly claims the same result also for Biclique Partition. The bug is acknowledged here: \url{https://sites.google.com/site/parinyachalermsook/research?authuser=0}}.
Edge Clique Cover is hard to approximate within $n^{0.5-\varepsilon}$ due to a reduction by \citet{Kou78}. In contrast, a PTAS exists for Edge Clique Cover on 
planar graphs~\cite{blanchette2012clique}.

\subsection{Preliminaries}
\label{sub:preliminaries}
For an $m\times n$ matrix $A$, we use $A_{i,j}$ to denote the entry of $A$ at row $i$ and column $j$.
We use $A_{i}$ to denote the row-vector given by the $i$-th row of $A$.
For some $I\subseteq [m]$ and $J\subseteq [n]$, we use $A_{I,J}$ to denote the sub-matrix of $A$ when restricted to rows with indices in $I$ and columns with indices in $J$.
Also, we use $A_I$ to denote a sub-matrix of $A$ when restricted to rows with indices in $I$. 
We call such a sub-matrix where only rows are restricted as row sub-matrix.
A \emph{row-basis} (or just \emph{basis} for brevity) $B$ of $A$ is any row sub-matrix of $A$ such that every row of $A$ can be expressed as a linear combination of rows of $B$, and the rows of $B$ are linearly independent with each other.

\begin{lemma}
\label{lem:equiv}
AWECP is equivalent to BSD-DW.
\end{lemma}
\begin{proof}
	Given an instance $(G,w_e,W,w_v,k)$ of AWECP, we can construct an instance of $(A,k)$ of 
BSD-DW as follows. Let $V(G)=\{1,\ldots n\}$; take the non-diagonal entries of 
$A$ as the corresponding entries of the \emph{weighted adjacency matrix} of $G$, 
i.e., if there is an edge between two vertices $u$ and $v$,
the entry $A_{u,v}$ 
is equal to $w_e(uv)$ and if $u$ and $v$ do not have an edge between them then $A_{u,v}=0$; for every vertex $v\in W$, take $A_{v,v}$ as 
the vertex weight of $v$;  for every vertex $v\in V(G)\setminus W$, take 
$A_{v,v}$ as the wildcard $\wild$. Note that the mapping is invertible, i.e., 
given a BSD-DW instance $(A,k)$ we get an AWECP instance $(G,w_e,W,w_v,k)$ as follows. 
Take $V(G):=\left\{ 1,2,\cdots,n \right\}$ where $n$ is the number of rows (and 
columns) of $A$. For distinct $u,v\in [n]$, if $A_{u,v}$ is non-zero, put an 
edge between $u$ and $v$ in $G$ with weight $A_{u,v}$. For each $v\in [n]$ such 
that $A_{v,v}$ is not a wildcard, put $v$ in $W$ and set its vertex weight to 
$A_{v,v}$. It is clear that this mapping is a bijective mapping between AWECP 
and BSD-DW instances.

Now, we define a bijective mapping between candidate solutions of the two 
problems. Naturally, a candidate solution of AWECP is a set of $k$ cliques and a 
candidate solution of BSD-DW is an $n\times k$ matrix. Consider a candidate 
solution $\calC:=\left\{ C_1,C_2,\cdots C_k \right\}$ of an AWECP instance 
$(G,w_e,W,w_v,k)$. We map it to a candidate solution $B\in\bin^{n\times k}$ of a BSD-DW 
instance $(A,k)$ as follows. Take the row $B_u$ as the characteristic vector of 
$u$ in the $k$ cliques, i.e., $B_{u,j}:=1$ if $u\in C_j$, and $B_{u,j}:=0$ 
otherwise. The inverse mapping then turns out to be as follows. Given a 
candidate solution $B\in\bin^{n\times k}$ of instance $(A,k)$ construct $k$ 
cliques where the $j$-th clique is $C_j:=\left\{ u\mid B_{u,j}=1 \right\}$. To 
see that $C_j$ is indeed a clique, consider any two vertices $u,v\in C_j$: 
since 
$B_{u,j}=B_{v,j}=1$, we know that $A_{u,v}=B_uB_v^T\ge 1$, which implies that 
there is an edge between $u$ and $v$ in $G$.

First, we prove that if $\calC$ is a solution of 
AWECP$(G,w_e,W,w_v,k)$, then 
 $B$ is a solution of BSD-DW$(A,k)$. It is clear that $B$ has only $k$ columns by construction. So, it only remains to 
prove that for all pairs $u,v\in [n]$, $B_uB_v^T\es A_{u,v}$. First consider 
the 
case when $u$ and $v$ are distinct. Let $J$ denote the set of all $j$ such that 
both $u$ and $v$ appear together in $C_j$. Since $\calC$ is a solution of 
AWECP$(G,w_e,W,w_v,k)$, we have that $|J|=A_{u,v}$. By construction of $B$, we have that 
$J$ is exactly the set of indices $j$ where $B_{u,j}=B_{v,j}=1$. Thus 
$B_uB_v^T=|J|=A_{u,v}$. Now consider the case when $u=v$. If $A_{u,u}$ is a 
$\wild$ then clearly $B_uB_u^T \es \wild=A_{u,u}$. So, suppose $A_{u,u}\neq 
\wild$. This means $u\in W$ implying that $u$ appears in exactly $A_{u,u}$ many 
cliques in $\calC$. Thus $B_uB_u^T=A_{u,u}$.

	It only remains to prove that if $B$ is a solution of BSD-DW$(A,k)$,
	then $\calC$ is a solution of $(G,w_e,W,w_v,k)$, which we do now.
By construction, $\calC$ has at most $k$ cliques.
Thus, it is 
sufficient to prove the following two statements: (1) every pair $u,v\in V(G)$ 
appears together in exactly $A_{u,v}$ many cliques in $\calC$ (2) each vertex 
$v\in W$ appears in $A_{v,v}$ many cliques in $\calC$. First we prove (1). We 
know $B_uB_v^T=A_{u,v}$. Since $B$ is binary, this means that there are exactly 
$A_{u,v}$ many indices $j$ such that $B_{u,j}$ and $B_{v,j}$ are both~$1$. Let 
$J$ be the set of those indices. Observe that the set of cliques where both $u$ 
and $v$ appear together are exactly $\{C_j:j\in J\}$. Thus, the edge $uv$ is in 
$|J|=A_{u,v}$ many cliques. Now we prove (2). Consider a vertex $v\in W$. We 
know $B_vB_v^T=A_{v,v}$. Since $B$ is binary, this means that there are exactly 
$A_{v,v}$ many ones in $B_{v}$. Thus, the vertex $v$ is in $A_{v,v}$ many 
cliques.
\end{proof}

\section{Kernel} 
\label{sec:kernel}
We will now give a kernel for AWECP and BSD-DW, thereby proving Theorem~\ref{thm:kernel}.
Let $G$ be the input graph to AWECP and $A$ be the corresponding input matrix to BSD-DW 
obtained by the transformation as in the proof of \cref{lem:equiv}.
We may move seemlessly between the graph and matrix terminologies as both problems are 
equivalent.
Whenever we say a solution in this section, we mean the solution to the BSD-DW instance i.e., a rank-$k$ BSD of $A$.
We say two distinct vertices $u$ and $v$ are {\bf twins} if they are adjacent and satisfy $A_u\es 
A_v$.
We first prove the following easy property of twins.

\begin{lemma}
	\label{lem:trans}
	For distinct vertices $u,v$ and $w$, suppose $u$ and $v$ are twins and $v$ and 
	$w$ are twins. Then: 
	\begin{enumerate}
		\item $u$ and $w$ are twins, and
\item
	all the entries of the submatrix $A_{\{u,v,w\},\{u,v,w\}}$ are the same 
except for wildcards.
	\end{enumerate}
\end{lemma}
\begin{proof}
First, let us prove the second statement. Let $A_{u,v}=\alpha$. Then we know 
$A_{u,w}=\alpha$ as $v$ and $w$ are twins. Then $A_{v,w}=\alpha$ as $u$ and $v$ 
are twins. Thus all the non-diagonal elements of $A_{\{u,v,w\}\{u,v,w\}}$ are 
equal to $\alpha$. If $A_{u,u}\neq \wild$ then $A_{u,u}=A_{v,u}=\alpha$ as $u$ 
and $v$ are twins. Similarly, if $A_{v,v}\neq \wild$ then 
$A_{v,v}=A_{v,u}=\alpha$ as $u$ and $v$ are twins. And, if $A_{w,w}\neq \wild$ 
then $A_{w,w}=A_{v,w}=\alpha$ as $v$ and $w$ are twins.

Now, for the first statement to hold, we only need to show that 
$A_{u,z}=A_{w,z}$ for all $z\notin \left\{ u,v,w \right\}$.
Indeed, $A_{u,z}=A_{v,z}=A_{w,z}$ where the first equality is because $u$ and 
$v$ are twins and the second is 
because $v$ and $w$ are twins.
\end{proof}

Thus we have that the relation \emph{twins} is transitive. It is also symmetric, as easily seen from the definition.
To make it also reflexive, we consider a vertex to be twin of itself.
Thus, we can group the vertices into equivalence classes of twins.
We call each equivalence class a {\bf block}.
Note that there can be blocks containing only a single vertex.
The following lemma is a direct consequence of \cref{lem:trans}.
\begin{lemma}
	\label{lem:block-same}
	For a block $D$, the entries of the sub-matrix $A_{D,D}$ are all same 
except for wildcards.
\end{lemma}

\begin{fact}
	\label{fact:es}
	For values $a,b$ and $c$, if $a\es b$ and $b\es c$, and $b\neq \wild$ then 
$a\es c$.
\end{fact}
\begin{lemma}
Suppose we have a YES instance of AWECP without isolated vertices.
Then there can be at most $2^k$ blocks.
\end{lemma}
\begin{proof}
	Let $B$ be a rank-$k$ BSD of $A$.
	Note that $B$ exists as we have a YES instance.
	In order to prove the lemma, it is sufficient to show that if $u$ and $v$ are in different blocks, then 
$B_u$ and $B_v$ are distinct, because then there can only be $2^k$ distinct rows of 
$B$, as there are only $k$ columns in $B$ and $B$ is binary. We will prove the 
contrapositive, i.e., we wil show that if $B_u=B_v$ then $u$ and $v$ are in the same 
block. 
Assume for the sake of contradiction that $B_u=B_v$ and $u$ and $v$ are in different blocks, 
i.e., they are not twins. Let $b:=B_uB^T=B_vB^T$. We have $A_u\es 
B_uB^T= b$ and $A_v\es B_vB^T=b$. This implies $A_u\es A_v$ 
using 
\cref{fact:es}, as the vector $b$ contains no wildcards. Then, for $u$ 
and $v$ to be not twins, it should be the case that $u$ and $v$ are not 
adjacent, i.e, $A_{u,v}=0$. But then, $B_{u}B_{v}^T=0$. 
Since $B_u=B_v$ by assumption, we have that 
 $B_u=B_v=\mathbf{0}$ and hence $A_u=A_v=\mathbf{0}$. 
 This means that $u$ and $v$ are isolated vertices, which is a contradiction.
\end{proof}

The above lemma shows the soundness of our first reduction rule that is as follows.
\begin{reduction}
\label{rul:blocks}
If the number of blocks is more than $2^k$, output that the instance is a NO 
instance. 
\end{reduction}

Next, we prove the following lemma about twins that helps us to come up with a reduction rule that bounds the size of each block.

\begin{lemma}\label{lem:alldist}
Let $D:=\{v_1,v_2,\dots,v_t\}$ be a block of twins. For a YES instance, there 
exists a solution $B$ such that the rows $B_{v_1},B_{v_2},\dots B_{v_t}$ are 
either all pairwise distinct, or all same.
\end{lemma}
\begin{proof}
It is sufficient to prove the following statement: if there is a solution $B$ 
such that $B_{v_1}=B_{v_2}$, then there is also a solution $C$ such that 
$C_{v_1}=C_{v_2}=\dots = C_{v_t}$. So, assume that $B_{v_1}=B_{v_2}$. 
Let $C$ be the matrix defined as $C_{v}:=B_{v}$ for all $v\notin D$, and 
$C_{v}:= B_{v_1}=B_{v_2}$ for all $v\in D$. We will prove that $C$ is also a 
solution. For this, it is sufficient to prove that $C_{u}C_{v}^T=A_{u,v}$ 
for all $u,v\in V$ such that $A_{u,v}\neq \wild$. If both $u$ and $v$ are not 
in $D$, then $C_{u}C_{v}^T=B_{u}B_{v}^T=A_{u,v}$. So, without loss of 
generality assume that $u\in D$. We distinguish the following cases.
\begin{enumerate}
 \item 	If $v\in V\setminus D$, then 
$C_{u}C_{v}^T=B_{v_1}B_{v}^T=A_{v_1,v}=A_{u,v}$, where the last equality 
follows as $v_1$ and $u$ are twins.
\item If $v\in D\setminus\{u\}$, then
$C_{u}C_{v}^T=B_{v_1}B_{v_2}^T=A_{v_1,v_2}=A_{u,v}$, where the last 
equality follows from \cref{lem:block-same}. 
\item If $v=u$: if $A_{u,u}=\wild$ then there is nothing to prove, so assume 
$A_{u,u}\neq \wild$. Then $A_{u,u}=A_{v_1,v_2}$ by \cref{lem:block-same}. 
Hence we get $C_{u}C_{u}^T=B_{v_1}B_{v_2}^T=A_{v_1,v_2}=A_{u,u}$.\qedhere
\end{enumerate}
\end{proof}

Since there are 
only $2^k$ possible distinct rows for a solution $B$, \cref{lem:alldist} has the following consequence.

\begin{lemma}\label{lem:allsame}
Let $D:=\{v_1,v_2,\dots,v_t\}$ be a block of twins such that $t>2^k$. For a YES instance, there 
exists a solution $B$ such that the rows $B_{v_1},B_{v_2},\dots B_{v_t}$ are 
all same.
\end{lemma}

The above lemma suggests that for a block $D$ of size more than $2^k$, we only need to keep 
one representative vertex for all the vertices in $D$.
This leads us to our second reduction rule.
\begin{reduction}\label{rul:block-size}
Suppose there is a block $D$ with more than $2^k$ vertices. Pick any two 
arbitrary vertices $u,v\in D$. 
We reduce our instance to an  instance $A'$ of AWECP (simulataneously to an instance $G'$ of BSD-DW) as follows: 
let $G':=G\setminus (D\setminus \left\{ v \right\})$;
for every pair $(v_1,v_2)\neq (v,v)$ in $V(G')\times V(G')$, let $A'_{v_1,v_2}:=A_{v_1,v_2}$; 
let $A'_{v,v}:=A_{u,v}$.

Once we have a solution $B'$ to the reduced instance $A'$ then we construct a solution $B$ to the original instance $A$ as follows: 
for all $x\in D$, let $B_x:=B'_v$; for all $x\in V(G)\setminus D$, let $B_x:=B'_x$.
\end{reduction}
Now, we prove that the above reduction rule is safe.
\begin{lemma}
	\label{lem:safeness}
	Let $A',G',B',B$ be as defined in \cref{rul:block-size}. 
	\begin{enumerate}
		\item
	If $B'$ is a rank-$k$ BSD of $A'$, then $B$ is a rank-$k$ BSD of $A$.
\item
	Conversely, if $A$ has a rank-$k$ BSD then so does $A'$.
	\end{enumerate}
\end{lemma}
\begin{proof}
\begin{enumerate}
	\item	
		It is clear that $B$ has only $k$ columns. So, it only remains to prove that $B$ is a BSD of $A$, for which it is sufficient to prove that $B_{v_1}B_{v_2}^T\es A_{v_1,v_2}$ for all $v_1,v_2\in V(G)$. 
		For $v_1,v_2 \in V(G)\setminus D$, we have 	
		\begin{align*}
			B_{v_1}B_{v_2}^T=B'_{v_1}B'^T_{v_2}\es A'_{v_1,v_2}=A_{v_1,v_2}.
		\end{align*}
		For $v_1\in V(G)\setminus D$ and $v_2\in  D$, we have 	
		\begin{align*}
			B_{v_1}B_{v_2}^T=B'_{v_1}B'^T_{v}=A'_{v_1,v}=A_{v_1,v}=A_{v_1v_2},
		\end{align*}
		where the last equality follows as $v$ and $v_2$ are twins.

		For $v_1,v_2\in  D$, we have 	
		\begin{align*}
			B_{v_1}B_{v_2}^T=B'_{v}B'^T_{v}=A'_{v,v}=A_{u,v}=A_{v_1,v_2},
		\end{align*}
		where the last equality follows from Lemma~\ref{lem:block-same}.
	\item
		By \cref{lem:allsame} we know that there exists a rank-$k$ BSD of $A$ such that $B_{v_1}=B_{v_2}$ for all $v_1,v_2\in D$. 
		In particular $B_u=B_v$.
		Let $B'$ be defined as 
		$B'_x:=B_x$ for all $x\in V(G')$.
		We show that $B'$ is a rank-$k$ BSD of $A'$.
		Since $B'$ has only $k$ columns, it only remains to prove that $B'$ is a BSD of $A'$, which we do as follows.
		For $(v_1,v_2) \in \left(V(G')\times V(G')\right)\setminus (v,v)$, we have
		\begin{align*}
			B'_{v_1}B'^T_{v_2}=B_{v_1}B^T_{v_2}\es A_{v_1,v_2}=A'_{v_1,v_2}.
		\end{align*}
		And,
		\begin{align*}
			B'_{v}B'^T_{v}=B_{v}B^T_{v}=B_{u}B^T_{v}=A_{u,v}=A'_{v,v}.
		\end{align*}
\end{enumerate}	
\end{proof}

After the above rules are exhaustively applied, each block has size at 
most~$2^k$ and the number of blocks is at most $2^k$. Thus we have the required 
kernel of size~$4^{k}$. The time required for computing the kernel can be 
shown to be~$\bigO(n^2\log n)$. This is achieved by using sorting to find 
blocks of twins.
Since the edge weights and vertex weights for vertices in~$W$ cannot exceed $k$ 
if there is a solution with at most $k$ cliques, a kernel with at most $4^k$ 
vertices can be encoded using $\bigO(\binom{4^k}{2}\log k)$ bits, and so 
\cref{thm:kernel} follows.

\section{Algorithm} 
\label{sec:algorithm}

Here we give an algorithm for the BSD-DW problem.
The algorithm also solves AWECP due to the equivalence from Lemma~\ref{lem:equiv}.
In particular, it solves WECP thereby proving Theorem~\ref{thm:algo}.

We now give a description of the algorithm.
A pseudocode is given in \cref{alg:bsd}.
Our input is a symmetric matrix $A\in \zs^{n\times n}$ where wildcards $\wild$
appear only on the diagonal. 
First we guess a matrix $P\in\binkk$ such that for some $r\le k$, $P_{[r],[k]}$ is a row basis of solution $B$. 
We show that for this, it is sufficient to enumerate $k\times k$ binary matrices that satisfy a specific property defined as follows.
Let $w$ be the largest integer entry of $A$.
We call a matrix $\mathbf{w}${\bf -limited} if the dot-product of each pair of 
its rows is at most~$w$. 
The following fact shows that we only need to enumerate $w$-limited matrices in $\binkk$ to guess $P$. 
\begin{fact}
	\label{fact:sol-wlimited}
	If $B$ is a BSD of matrix $A$ and $w$ is the largest integer entry of $A$, then 
any submatrix of $B$ (including~$B$) is $w$-limited.	
\end{fact}
Note that guessing $P$ is done in Loop 1 of Algorithm~\ref{alg}.
We will later give a bound on the number of $w$-limited matrices in $\binkk$ during the runtime analysis in Section~\ref{sub:runtime_analysis}, thereby bounding the number of iterations of Loop 1.

We maintain partially filled matrices during the algorithm, i.e., we allow 
matrices to have \emph{null rows} (this is different from wildcards). Think of 
the null rows as the rows that have not been filled yet. If each row of a matrix 
is either a binary row or a null row, we call it a \emph{binary matrix with 
possibly null rows}. We denote by~$\B^{n\times k}$, the set of all $n\times k$ 
binary matrices with possibly null rows. 

We maintain a matrix $\basis\in \Bnk$ as a potential basis for our solution $B$.
In Line~\ref{line:bext-call}, we call $\ebas$ that checks whether the current 
$\basis$ can be extended to a full solution $B$. Note that $\ebas$ does not try 
all possibilities to fill the remaining rows. It fills a row with the first 
binary vector that is compatible with the rows so far, where compatibility is 
defined as follows. For a matrix $B\in \Bnk$, we say that a vector $v\in \bin^k$ 
is \mbox{$\mathbf{i}${\bf -compatible}} for $B$ if $v^Tv \es A_{i,i}$, and for 
all $j\neq i$ such that $B_{j}$ is not a null row, $v^TB_{j}^T=A_{i,j}$. If 
$\ebas$ is able to fill all the rows with $i$-compatible binary vectors, then we 
are done and we return the resulting matrix (in Line~\ref{line:yes-output}). If 
not, we claim that the row for which we are not able to fill can be added 
to the basis (in Claim~\ref{claim:lin-ind-sub}). So we add one more row to the 
basis by copying the next row from $P$ (in Line~\ref{line:basis-inc}). Thus we 
increase the number of non-null rows in the basis $\basis$ by one and repeat. 
Since the basis can be at most of size $k$, we need to repeat this at most $k$ 
times.

\begin{algorithm}[tb]
	\caption{Algorithm for BSD-DW\label{alg:bsd}}	
	\label{alg}
\DontPrintSemicolon
	\SetKwData{Bsz}{basis\_index}
	\SetKwData{Basis}{Basis}
	\SetKwData{Comp}{compatibility}
	\SetKwData{True}{true}
	\SetKwData{False}{false}
	\SetKwFunction{Ebas}{ExtendBasis}
	\setcounter{AlgoLine}{0}
	\SetKwInOut{Input}{Input}\SetKwInOut{Output}{Output}
	\SetKwInOut{Assumption}{Assumption}
	
	\Input{An $n\times n$ symmetric integer diagonal-wildcard matrix $A$}
	\Output{If $A$ has a rank-$k$ BSD then output a rank-$k$ BSD $B$ of $A$; \\
			otherwise report that $A$ has no rank-$k$ BSD}
	\BlankLine
	
	$w\leftarrow$ largest integer weight in $A$\;
	\ForEach(\tcp*[f]{Loop~1}) {w-limited $P\in \bin^{k\times k}$ } {
		Initialize $\basis$ to be an $n\times k$ matrix with all null rows 
\label{line:init}\;
		$b\leftarrow 1$\;
		$i\leftarrow 1$\;
		\While(\tcp*[f]{Loop~2}){$b\le k$ and $P_{b}$ is 
$i$-compatible with $\basis $}{
			$\basis_{i}\leftarrow P_{b}$\label{line:basis-inc}\;
			$(B,i)\leftarrow \Ebas(A,\basis)$\label{line:bext-call}\;
			\lIf {$i=n+1$}{ {\bf output} $B$ and {\bf terminate} the 
algorithm\label{line:yes-output}} 
			$b\leftarrow b+1$ \label{line:basis-size-inc}\;
		}	
	}
	{\bf output} that $A$ has no rank-$k$ BSD and {\bf terminate} the 
algorithm \label{line:no-output} \;
	\SetKwProg{Fn}{Function}{:}{end of function}
	\let\oldnl\nl%
	\newcommand{\nonl}{\renewcommand{\nl}{\let\nl\oldnl}}
	\nonl \;
	\nonl \Fn{\Ebas{A,B}}{
		\For(\tcp*[f]{Loop~3}){each null row $i$ in $B$ in increasing order  
\label{line:while-bext}}{
			\If {there is a $v\in\bin^k$ such that $v$ is $i$-compatible with 
$B$\label{line:condition-icomp}}{
					$B_{i}\leftarrow v$ \label{line:fillrow2}	
			}
			\lElse{\KwRet{$(B,i)$}}\label{line:out-i}
		}
		\KwRet{ $(B,n+1)$}\label{line:out-bout}
	}
\end{algorithm}

\subsection{Correctness of the algorithm}
\label{sub:correctness_of_the_algorithm}

The algorithm %
outputs either through Line~\ref{line:yes-output} or through 
Line~\ref{line:no-output}. In the former case, we prove the following claim.
\begin{claim}%
	\label{claim:out-bout}
	If output occurs through Line~\ref{line:yes-output}, then the matrix $B$ 
that is output, is a rank-$k$ BSD of $A$.
\end{claim}
\begin{proof}%
If Line~\ref{line:yes-output} is executed, then this means that the preceding 
\Ebas call on Line~\ref{line:bext-call} returned $i=n+1$. This implies that the 
return from \Ebas happened on Line~\ref{line:out-bout}. This in turn means that 
the condition of the while loop in 
Line~\ref{line:while-bext} was no longer true. This means the matrix $B$ did not 
have any null rows at the time of return. Thus $B\in \bin^{n\times k}$. The rows 
of $B$ were each filled either in Line~\ref{line:basis-inc} (when it was 
$\basis$ before being passed to \Ebas) or in Line~\ref{line:fillrow2}. In both 
places, we filled each row $i$ with a vector that was $i$-compatible at the time 
of filling. From the definition of $i$-compatibility, it follows that $BB^T\es 
A$, and hence $B$ is a rank-$k$ BSD of $A$.
\end{proof}

Consider a NO instance first. From \cref{claim:out-bout} it follows that the 
output does not occur through Line~\ref{line:yes-output}.
Thus the output has to occur through Line~\ref{line:no-output} and hence we 
correctly output that $A$ does not have a rank-$k$ BSD.
So it only remains to prove the correctness when $A$ is a YES instance, 
i.e., when $A$ has a rank-$k$ BSD, which is the case we consider for the 
remainder of the proof. Let $B^*$ be any fixed rank-$k$ BSD of $A$.

Observe that $\basis$ changes as follows during each iteration of Loop~1: it is 
initialized to all null rows and each time the algorithm encounters 
Line~\ref{line:basis-inc} a null row is replaced with a binary row vector. We 
say that a matrix $B$ is {\bf consistent} with $B^*$ if $B_{j}=B^*_{j}$ for each 
$j$ such that $B_{j}$ is a non-null row.

\begin{claim}
	\label{claim:lin-ind-sub}
Consider a matrix $\basis\in \Bnk$ that is consistent with $B^*$.
If $\Ebas(A,\basis)$ returns $i\in [n]$ then $B^*_{i}$ is linearly independent 
from the non-null rows of $\basis$. 
\end{claim}
\begin{proof}
For a matrix $M\in \Bnk$, we denote by~$\nnr(M)$ the set of indices of the 
non-null rows of $M$. Suppose for the sake of contradiction that 
$\Ebas(A,\basis)$ returns $i\in [n]$ and $B^*_i$ is linearly dependent on the 
non-null rows of $\basis$. Then, we have $B^*_i=	\Sigma_{\ell\in 
\nnr(\basis)}\lambda_{\ell}\basis_{\ell}$ for some 
$\lambda_1,\lambda_2,\cdots,\lambda_{\ell}\in \RR$. Since $\basis$ is consistent 
with $B^*$, we can write $B^*_i=	\Sigma_{\ell\in 
\nnr(\basis)}\lambda_{\ell}B^*_{\ell}$.

As \Ebas returned $i$, we know that during that iteration of Loop~3 in which row 
$i$ was considered, no vector $v\in \bin^k$ was $i$-compatible with~$B$ (here $B$ 
is the matrix maintained by \Ebas that was initialized to $\basis$ by the 
function call). In particular, $B^*_i\in \bin^k$ was not $i$-compatible with 
$\basis$. Therfore, either there was some $j\in \nnr(B)$ such that 
$B^*_iB_{j}^T\neq A_{i,j}$, or $B^*_i(B^*_i)^T\not\es A_{i,i}$. The latter 
cannot be true as $B^*$ is a rank-$k$ BSD of~$A$. So there was a $j\in \nnr(B)$ 
such that $B^*_iB_{j}^T \neq A_{i,j}$.

We branch into two cases: case 1 when $j\in \nnr(\basis)$ and case 2 when $j\in 
\nnr(B)\setminus \nnr(\basis)$. In case 1, we have $B_{j}=\basis_{j}=B^*_{j}$ 
where the second equality is because $\basis$ and $B^*$ are consistent. Thus 
$B^*_iB_{j}^T = B^*_i(B^*_{j})^T = A_{i,j}$, giving a contradiction.

In case 2, $B_j$ was added in Line~\ref{line:fillrow2} and hence $B_j$ was 
$j$-compatible with $B$ at this time, implying that $B_{\ell}B_j^T=A_{\ell, j}$ 
for all $\ell \in R(\basis)$. Since $B_{\ell}=\basis_{\ell}=B^*_{\ell}$ for 
$\ell \in R(\basis)$, we have that $B^*_{\ell}B_j^T=A_{\ell, j}$ for all $\ell 
\in R(\basis)$. Then, we have
\begin{align*}
	B^*_iB_{j}^T & = \Sigma_{\ell\in 
\nnr(\basis)}\lambda_{\ell}B^*_{\ell}B_j^T\\
				 & = \Sigma_{\ell\in \nnr(\basis)}\lambda_{\ell}A_{\ell, j}\\
				 & = \Sigma_{\ell\in 
\nnr(\basis)}\lambda_{\ell}B^*_{\ell}(B^*_j)^T\\
				 & = B^*_{i}(B^*_j)^T\\
				 & = A_{i,j}
\end{align*}
This is a contradiction.
\end{proof}

For a matrix $X\in \binkk$, we say we are in iteration $(X,t)$ of 
the algorithm if we are in the iteration of Loop~1 with $P=X$ and the iteration 
of Loop~2 with $b=t$. We use $\basis(X,t)$ to denote the value of $\basis$ 
after the execution of Line~\ref{line:basis-inc} during iteration $(X,t)$. 

\begin{claim}%
	\label{claim:lin-ind}	
	At any step of the algorithm, if $\basis$ is consistent with $B^*$ 
then the non-null rows of $\basis$ are linearly independent.
\end{claim}
\begin{proof}%
Consider the first time this is violated during the algorithm. This has to be 
during the addition of a new non-null row at Line~\ref{line:basis-inc}. Let 
$(X,t)$ be the iteration in which this happens. Let $p$	be the index of the row 
that was added. Observe that $\basis(X,t)$ has only one additional non-null row 
compared to $\basis(X,t-1)$. Also, this additional non-null row is equal to 
$B^*_p$ as $\basis(X,t)$ is consistent with~$B^*$. We know the rows of 
$\basis(X,t-1)$ are linearly independent as we assumed that the first violation 
of lemma happens in iteration $(X,t)$. Also, during iteration $(X,t-1)$, $i$ was 
returned with value $p$ (as the insertion happens in Line~\ref{line:basis-inc} 
in iteration $(X,t)$). This implies that $B^*_p$ is linearly independent from 
the non-null rows of $\basis(X,t-1)$ due to \cref{claim:lin-ind-sub}. Hence the 
rows of $\basis(X,t)$ are linearly independent.
\end{proof}

\begin{claim}
\label{claim:casek}
If the iteration $(X,k)$ occurs during the algorithm for some $X\in \binkk$ such 
that $\basis(X,k)$ is consistent with $B^*$ then the algorithm outputs through 
Line~\ref{line:yes-output} in iteration~$(X,k)$.
\end{claim}
\begin{proof}
	Consider the $i$ returned by $\Ebas(A,\basis(X,k))$. It is sufficient to 
	prove that the condition $i=n+1$ in Line~\ref{line:yes-output} is satisfied.
	Suppose otherwise.
	Then $i\in [n]$ and by \cref{claim:lin-ind-sub}, $B^*_i$ is linearly independent from the 
non-null rows of $\basis(X,k)$. But by \cref{claim:lin-ind}, we have that the 
non-null rows of $\basis(X,k)$ are linearly independent and hence span the whole 
space, thus giving a contradiction.
\end{proof}

\begin{claim}
	\label{claim:consistent}
	Assume that the output of the algorithm does not occur through 
Line~\ref{line:yes-output}. If for some $Y\in \bin^{k\times k}$ and $t\le k-1$, 
iteration $(Y,t)$ occurs and $\basis (Y,t)$ is consistent with $B^*$, then there 
exists some $Z\in \bin^{k\times k}$ such that iteration $(Z,t+1)$ occurs and 
$\basis(Z,t+1)$ is consistent with $B^*$. 
\end{claim}
\begin{proof} 
Since $\basis (Y,t)$ is consistent with $B^*$, we know that $Y_{[t]}$ is a 
sub-matrix of $B^*$. As the condition in Line~\ref{line:yes-output} is false, 
we know that an $i\in [n]$ was returned in Line~\ref{line:bext-call} in 
iteration~$(Y,t)$. It is clear from the algorithm that $i$ is a null-row in 
$\basis (Y,t)$. Let $Z\in \binkk$ be such that $Z_{[t]}:=Y_{[t]}$, 
$Z_{t+1}:=B^*_i$, and $Z_{q}:=\zerov$ for all $q\geq t+1$. Observe that 
$Z_{[t+1]}$ is a submatrix of $B^*$ and hence is $w$-limited by 
\cref{fact:sol-wlimited}. Since adding zeroes does not destroy $w$-limitedness, 
we have that $Z$ is a $w$-limited $n\times k$ matrix. Thus there is some 
iteration of Loop~1 with $P=Z$. In this iteration the algorithm behaves 
similarly to the iteration with $P=Y$ for the first $t$ iterations of Loop~2 as 
the algorithm has seen only the first $t$ rows of $P$ up to then. Thus $\basis 
(Z,t)=\basis (Y,t)$ and $i$ is returned by Line~\ref{line:bext-call} in 
iteration $(Z,t)$. Now in Line~\ref{line:basis-inc} of iteration~$(Z,t+1)$, 
$\basis_{i}$ is assigned $Z_{t+1}$. Note that $Z_{t+1}=B^*_i$ is indeed 
$i$-compatible with $\basis(Z,t)$ (as $\basis(Z,t)=\basis(Y,t)$ and $\basis(Y,t)$ 
is consistent with $B^*$) and that $t+1\le k$. Hence the loop condition of 
Loop~2 is true in iteration $(Z,t+1)$. Thus, we have $(\basis (Z,t+1))_{i}= 
Z_{t+1}=B^*_i$ and for all $j\neq i$, we have $(\basis (Z,t+1))_{j}=(\basis 
(Y,t))_{j} $. Since $\basis (Y,t)$ is consistent with $B^*$, it follows that 
$\basis (Z,t+1)$ is consistent with $B^*$.
\end{proof}

Let $t$ be the largest number for which there exists a $P\in \binkk$ such that 
iteration~$(P,t)$ happens and $\basis(P,t)$ is consistent with $B^*$. Due to 
\cref{claim:consistent}, we know that $t=k$. Then the algorithm outputs through 
Line~\ref{line:yes-output} according to \cref{claim:casek}. Thus the algorithm 
outputs a correct solution $B$ due to \cref{claim:out-bout}.

\subsection{Runtime Analysis}
\label{sub:runtime_analysis}
First, let us bound the number of iterations of Loop 1. For this it is 
sufficient to bound the number of $w$-limited matrices in $\binkk$.

\begin{lemma}\label{lem:basisenum}
The number of binary $w$-limited $k\times k$ matrices is at most 
$(2e\sqrt{k/w})^{k^{3/2}w^{1/2}+k}$.
\end{lemma}
\begin{proof}
Note that no $w$-limited matrix can have a $2\times (w+1)$-sub-matrix having 
all ones. The number of ones in such a matrix is a special case of the
well-studied Zarankiewicz problem and is known~\cite{nikiforov2010contribution} 
to be at most $k^{3/2}w^{1/2}+k$. 
Hence it follows that the number of binary $w$-limited $k\times k$ matrices is 
at most $2^{k^{3/2}w^{1/2}+k}\cdot\binom{k^2}{k^{3/2}w^{1/2}+k}$ by choosing the 
positions of the at most $k^{3/2}w^{1/2}+k$ potential ones in the matrix and 
then choosing which of them are actually ones. The bound follows easily by using 
that $ \binom{n}{k}\le \left( \frac{ne}{k} \right)^k $. 
\end{proof} 

Next, let us analyse the runtime of the function \Ebas. Loop~3 has at most~$n$ 
iterations. In Line~\ref{line:condition-icomp}, we need to check at most $2^k$ 
vectors $v\in \bin^k$. The checking for $i$-compatibility of each vector takes 
$\bigO(nk)$ time. Hence \Ebas takes $\bigO(k2^kn^2)$ time.

Now, we are ready to calculate the total run time. Due to \cref{lem:basisenum}, 
Loop~1 has at most $(2e\sqrt{k/w})^{k^{3/2}w^{1/2}+k}$  iterations.
Line~\ref{line:init} takes $\bigO(nk)$ time. Loop~2 has at most $k$ iterations. 
Line~\ref{line:basis-inc} takes at most $\bigO(k)$ time. The call to \Ebas in 
Line~\ref{line:bext-call} takes at most $\bigO(k2^kn^2)$ time as we already 
calculated. Any other step takes only constant time. Thus the total running time 
is bounded by
$
\bigO\left(\left((2e\sqrt{k/w})^{k^{3/2}w^{1/2}+k}\right)\left(nk+k(k+k2^kn^2) 
\right) \right) 
=\bigO\left((2e\sqrt{k/w})^{k^{3/2}w^{1/2}+k}\cdot k^22^kn^2\right).
 $
We may run our algorithm on the kernel provided by \cref{thm:kernel}, which 
means we may set $n=4^{k}$ in the above expression. The total running time is
 \[
	 \bigO\left((2e\sqrt{k/w})^{k^{3/2}w^{1/2}+k}\cdot k^2 2^{5k}+n^2\log 
n\right)=(2e)^{(k^{3/2}w^{1/2}+\bigO(k))\log_{2e}(k/w)}+\bigO(n^2\log n).
 \]

\section{Lower bound for number of ones in basis matrix}\label{sec:fpp}

In this section we 
construct binary matrices for which there is a rank-$k$ BSD and every basis of every rank-$k$ BSD  
has $\Omega(k^{3/2})$ ones, thereby proving \cref{thm:ones}.

We obtain such instances via Finite Projective Planes (FPPs), which are defined 
by a set system $\calS$ over a universe $U$ of elements such that
\begin{enumerate}
	\item for each $e,e'\in U$ there is exactly one $S\in \calS$ containing 
both of them,
	\item for each $S,S'\in \calS$ there is exactly one $e\in U$ such that 
$e\in S\cap S'$, and
	\item there is a set of $4$ elements in $U$ such that no three of them are 
in any $S\in \calS$. 
\end{enumerate}

It is known~\cite{matouvsekinvitation} that the definition implies that both the 
number of elements and the number of sets are equal to $N^2+N+1$ for some $N\ge 
2$. Here $N$ is called the \emph{order} of the~FPP. It also follows that for an 
FPP of order $N$, each set has exactly $N+1$ elements and each element is 
contained in exactly $N+1$ sets. FPPs exist for every prime power~$N$. 
\begin{fact}
	\label{fact:prime-power}
	For every prime power $N$, there is an FPP of that 
order~\cite{matouvsekinvitation}.
\end{fact}
Given an FPP of order $N$, in the following we will denote the characteristic 
incidence matrix of elements and sets by $F\in\bin^{(N^2+N+1)\times(N^2+N+1)}$, 
where rows are elements and columns are sets.

We now give a reduction from the problem of finding an FPP of order $N$ to 
ECP. For this, consider a vertex set $V$ with $N^2+N+1$ vertices. Let $I$ be a 
subset of $N+1$ vertices in~$V$. Let $G_N$ be the graph defined as 
the clique over $V$ minus the clique over~$I$, i.e., every pair of vertices in $V$ is 
adjacent except when both are from $I$. In other words, if $X:=V\setminus I$, 
then $G_N$ is a split graph with $X$ as the clique and $I$ as the independent 
set, where all the adjacencies are present between $X$ and $I$.  The following 
lemmas show that $G_N$ has a small ECP solution if and only if an FPP of order 
$N$ exists.

\begin{lemma}\label{lem:FPP-ECP}
If a finite projective plane $\calS$ of order $N$ exists, then $G_N$ has a 
clique partition~$\calC$ into $|\calC|\leq N^2+N$ cliques. 
\end{lemma}
\begin{proof}
Let $\calS$ be an FPP of order $N$ over a universe $U$, and fix one of its sets 
$S\in\cal S$. We identify this set with the independent set of $G_N$, i.e., 
$S=I$. After fixing the elements of~$S$, all other elements in $U\setminus S$ 
are arbitrarily identified with the other vertices in $X$. We claim that the 
remaining sets in $\calS\setminus\{S\}$ form a clique partition, i.e., if 
$C_{S'}=\{uv\in E(G_N)\mid u,v\in S'\}$ then the set $\calC=\{C_{S'}\mid S'\in 
\calS\setminus\{S\}\}$ partitions the edge set of $G_N$ into cliques. From 
Property~1 of an FPP, for any edge $uv$ (i.e.,~at least one of $u$ and $v$ is in 
$X$) there is exactly one set $S'\in\calS\setminus\{S\}$ such that $u,v\in S'$. 
This means that the subgraphs in $\calC$ partition the edge set. Furthermore, by 
Property~2 no $S'\in\calS\setminus\{S\}$ intersects in more than one vertex with 
the independent set $I$. Thus every subgraph of $\calC$ is a clique. Moreover, 
any FPP of order~$N$ has exactly $N^2+N+1$ sets, and so there are $N^2+N$ 
cliques in $\calC$.
\end{proof}

To prove the other direction, i.e, that a small ECP solution to $G_N$ implies the 
existence of an~FPP, we need the following lemma.

\begin{lemma}
\label{lem:clique-size}
If $\calC$ is a set of cliques that partition the edges of $G_N$ and $|\calC|\le 
N^2+N$, then for each $C\in \calC$, $|V(C)|=N+1$.
\end{lemma}
\begin{proof}
First let us prove that $|V(C)|\le N+1$. Suppose for the sake of contradiction 
that $|V(C)|\ge N+2$. Note that $C$ contains at most one vertex from $I$, as a 
clique and independent set can intersect on at most one vertex. Let 
$C':=V(C)\setminus I$ and $I':=I\setminus V(C)$. Clearly $|C'|\ge N+1$ and 
$|I'|\geq N$ (recall that $|I|=N+1$). Note that every edge in $C'\times I'$ has 
to be covered by a distinct clique in $\calC\setminus \{C\}$: any two edges that 
have different endpoints in $I$ cannot be in the same clique, since there is no 
edge between these endpoints, while any two edges with different endpoints in 
$C$ cannot be in the same clique, since the only edge between these endpoints is 
already covered by $C$. But there are $|C'||I'|\ge N^2+N$ such edges whereas 
there are only $N^2+N-1$ cliques in $\calC\setminus \{C\}$. Thus we have a 
contradiction.

Hence we established $|V(C)|\le N+1$. Now suppose for the sake of contradiction 
$|V(C)|< N+1$. Using the fact that every clique of $\calC$ has size at 
most $N+1$, the total number of edges covered by $\calC$ is strictly less than 
$|\calC|\binom{N+1}{2}\le (N^2+N)\binom{N+1}{2}=N^2(N+1)^2/2$. However, since
$|I|=N+1$ and consequently $|X|=N^2$, the total number of edges of $G_N$ is 
$\binom{N^2}{2}+N^2\cdot (N+1)=N^2(N+1)^2/2$. Thus, we have a contradiction.
\end{proof}

Now, we prove the other direction.

\begin{lemma}\label{lem:ECP-FPP}
Let $N\geq 2$.
If $\calC$ is a set of cliques that partition the edges of $G_N$ such that
$|\calC|\le N^2+N$, then~$\calS=\{V(C)\mid C\in\calC\}\cup \{I\}$ is an FPP of 
order $N$ over $V$. Moreover, the incidence matrix~$F$ of $\calS$ with
the column for $I$ removed from it, is the BSD of the adjacency matrix of $G_N$ that corresponds to $\calC$. 
\end{lemma}
\begin{proof}
We will prove that $\calS=\{V(C)\mid C\in\calC\}\cup \{I\}$ satisfies the three 
properties in the definition of an FPP, which then has order $N$ by 
\cref{lem:clique-size}. Property~1 follows easily from the definition of an 
edge clique partition: for each pair of adjacent vertices there is exactly one 
clique covering their edge, while any pair of non-adjacent vertices only appear 
in $I$.

Let us now prove Property~2. For any $S,S'\in\calS$, it follows easily from 
the definition of an edge clique partition that $|S\cap S'|\le 1$ (otherwise 
some edge is contained in two cliques). Also, for any $S\in \calS$, it is true 
that $|S\cap I|\le 1$ (otherwise some clique would contain a non-edge). Assume 
there are $S,S'\in \calS$ with $S\cap S'=\emptyset$. By \cref{lem:clique-size}, 
we have $|S|=|S'|=N+1$, and so all the $(N+1)^2$ edges of $S\times S'$ have to 
be covered by distinct cliques (otherwise some clique would contain an edge 
already covered by one of the cliques induced by $S$ or~$S'$). But we do not 
have so many cliques as $|\calC|\leq N^2+N$. Thus we have $|S\cap S'|=1$ for 
any $S,S'\in \calS$, and so Property~2 is satisfied.

Let us now prove Property~3. Consider any arbitrary clique $C\in \calC$. Pick 
two vertices from $V(C)\setminus I$ and two vertices from $I\setminus V(C)$.  
Note that $|V(C)\setminus I|=|I \setminus V(C)|\geq N+1-1=N\ge 2$, and hence 
two vertices can be picked from the sets. It is easy to see that out of these 
four vertices at most two are in any set in $\calS$.

It is easy to see that the incidence matrix $F$ of $\calS$ minus the 
column for $I$ is the BSD of the adjacency matrix of $G_N$ that corresponds to the clique partition $\calC$.
\end{proof}

We use the above reduction from FPP to ECP, and the following fact about FPPs to prove \cref{thm:ones}.
\begin{fact}
	\label{fact:incidence}	
	The element-set incidence matrix of any FPP has full rank~\cite{sachar1979f}.
\end{fact}

\begin{proof}[Proof of \cref{thm:ones}]
	Let $N$ be a prime power
	and $k:=N^2+N$.
	We will show that the adjacency matrix $A$ of $G_N$ has a rank-$k$ BSD and every basis of every rank-$k$ BSD of $A$ has $\Theta(k^{3/2})$ ones.
	Note that $A$ is a $(k+1)\times (k+1)$ binary matrix as stated in the theorem.
	By Fact~\ref{fact:prime-power}, we have that there is an FPP of order $N$.
	Then by \cref{lem:FPP-ECP}, 
there is an edge clique partition of $G_N$ with at most $k=N^2+N$ cliques.
Thus, the adjacency matrix $A$ of $G_N$ has a rank-$k$ BSD, by using the equivalence in \cref{lem:equiv}.

Now, consider any rank-$k$ BSD $B$ of $A$ and
$\basis$ be any basis of $B$.
Then, by \cref{lem:equiv}, there is an edge clique partition of $G_N$ with at most $k$ cliques.
By \cref{lem:ECP-FPP}, $\calS=\{V(C)\mid 
C\in\calC\}\cup \{I\}$ is an FPP of order $N$.
Let $F$ be the element-set incidence matrix of $\calS$.
By \cref{lem:ECP-FPP}, $B$ is equal to $F$ 
 minus the column in $F$ corresponding to $I$. 
 By \cref{fact:incidence}, $F$ has full rank, i.e. it has rank $N^2+N+1=k+1$.
 This implies $B$ has rank $k$, and hence has at least $k$ columns. 
 Since $B$ is a rank-$k$ BSD, this means it has exactly $k$ columns, and hence is a $(k+1)\times k$ matrix.
 Since $B$ has rank $k$, we have that
 $\basis$ has $k$ rows and $k$ columns.
 Thus, $\basis$ is $B$ minus some row of $B$.
 Since each column of $B$ corresponds to a clique of $\calC$ containing $N+1$ vertices by 
 \cref{lem:clique-size}, we have that $B$ has $k(N+1)$ ones.
 Hence the number of ones in $\basis$ is at least $k(N+1)-k=\Theta(k\sqrt{k})$.
\end{proof}

\section{Conclusion and Open Problems}\label{sec:open_problems}
We showed that AWECP admits a kernel with $4^k$ vertices, and an algorithm with 
a  runtime of~$2^{O(k^{3/2}w^{1/2}\log(kw))}n^{O(1)}$, which implies that ECP 
can be solved in $2^{O(k^{3/2}\log k)}n^{O(1)}$ time. We think the following are the most interesting related open questions. 

\begin{itemize}
	\item Close the gap further between the upper and lower bounds on the running time for ECP that are currently $2^{\bigO(k^{3/2}\log k)}n^{\bigO(1)}$ and $2^{\Omega(k)}n^{\bigO(1)}$ respectively. 
	\item 
Does WECP admit a polynomial sized 
kernel like ECP?
	\item The algorithm of \citet{chandran2017parameterized} for Bipartite 
Biclique Partition with runtime $2^{O(k^{2})}n^{O(1)}$ is also based on guessing 
the basis of a binary decomposition $A=BC$, and is currently the fastest FPT algorithm for the problem.
If we can show that in any solution 
at least one of $B$ and $C$ has a row basis (column basis in case of $C$) with 
at most $g(k)$ ones, then we get a running time $2^{O(g(k)\log k)}n^{O(1)}$ 
using a similar algorithm as we gave for ECP. 
What 
is the minimum value of $g(k)$ possible?
	\item Can we show a tightness of analysis of our algorithm 
for WECP as we showed for ECP in \cref{sec:fpp}, i.e., can we construct 
positive integer matrices  with largest weight $w$ that has a rank-$k$ BSD and every basis of every rank-$k$ BSD have $\Omega(k^{3/2}w^{1/2})$ ones?
\end{itemize}

\printbibliography

\end{document}